\documentclass[a4paper,11pt]{article}
\usepackage[latin1]{inputenc}

\usepackage{fancybox,shadow,xcolor} 
\usepackage{setspace}
\usepackage{vmargin}
\setmarginsrb{2.5cm}{2.5cm}{2.5cm}{2.5cm}{0.5cm}{0.5cm}{0.5cm}{0.5cm}
\usepackage{interval}
\usepackage{xcolor}
\usepackage{psfrag}

\usepackage{graphicx}
\usepackage{amssymb,amsmath,amsthm}

\newtheorem{thm}{Theorem}
\newtheorem{lem}{Lemma}

\newtheorem{rem}{Remark}
\newtheorem{prop}{Proposition}

\theoremstyle{definition}
\newtheorem{definition}{Definition}
\newtheorem{assumption}{Assumption}

\usepackage{interval}

\DeclareMathOperator*{\rank}{rank}

\DeclareMathOperator*{\diag}{diag}

\DeclareMathOperator*{\tr}{tr}

\DeclareMathOperator*{\cl}{\mbox{\tiny\;cl}}
\DeclareMathOperator*{\sls}{\mbox{\tiny\;sls}}

\renewcommand{\Re}{\mathbb{R}}

\renewcommand{\paragraph}[1]{\smallskip\noindent\textbf{#1.} }

\newcommand{\BM}{\begin{bmatrix}}
\newcommand{\EM}{\end{bmatrix}}

\newcommand{\BBM}{\big[\begin{matrix}}
\newcommand{\EEM}{\end{matrix}\big]}

\newcommand{\bbm}{[\begin{matrix}}
\newcommand{\eem}{\end{matrix}]}

\newcommand{\rev}[1]{\textcolor{black}{#1}}  
   
\usepackage{titling}
\thanksmarkseries{arabic}

\title{Interval-valued estimation for discrete-time linear systems: application to  switched systems}

\author{Laurent Bako\thanks{L. Bako is with Universit\'{e} de Lyon, Amp\`{e}re (Ecole Centrale Lyon, INSA Lyon, Universit\'{e} Claude Bernard, CNRS UMR 5005), France. E-mail: laurent.bako@ec-lyon.fr.}  
 \hspace{1pt} and Vincent Andrieu\thanks{V. Andrieu is with Universit\'{e} de Lyon, LAGEP (Universit\'{e} Claude Bernard, CNRS UMR 5007), France. E-mail: vincent.andrieu@gmail.com}
}

\date{}

\begin{document}

\maketitle   
\begin{abstract}
\noindent This paper proposes a new framework for constructing interval-valued state estimators for discrete-time linear and switched linear systems.  Our main results are (i) the derivation of the tightest interval-valued estimator for linear discrete-time systems; (ii) a systematic design method for interval-valued state estimators; (iii) an application of the proposed estimation framework to switched linear systems.
\end{abstract}


\setstretch{1.2}

\section{Introduction}
Recovering the hidden state of  a dynamic process is a problem of major importance in many decision-making systems e.g., control or monitoring systems. State estimation methods rely on a mathematical model of the system which relates the unknown state to the  observed input and output of the system. However, often in practice, the model and the observed signals are uncertain and hence described by set-valued (in particular $\Re^n$ interval-valued) functions of time. As initiated in \cite{Gouze00-EM}, this form of uncertainty prompts the necessity of constructing interval-valued estimators which return at any time instant the set of all the possible values of the state. 
This type of estimators (observers) were investigated in a series of  papers for a variety of system classes:   continuous-time Linear Time Invariant (LTI) \cite{Bako18-CDC,Briat16-Automatica,Mazenc11-Automatica,Cacace15-TAC,Combastel13-TAC,Meslem16-CoDIT}, discrete-time linear systems \cite{Tang19-TAC,Efimov13-TAC,Mazenc14-IJNRC,Mazenc13-Automatica}, linear parameter-varying systems \cite{Chebotarev15-Automatica,Efimov13-TAC-b} and some specific classes of nonlinear systems \cite{Raissi12-TAC,Moisan10-SCL}. For more on the interval observer literature we refer to a recent survey reported in \cite{Efimov16-ARC}.

We consider in this paper the problem of designing interval-valued estimators for two classes of dynamic systems: discrete-time linear systems (LS) and switched linear systems (SLS). Concerning linear systems, many estimators exist in the literature as recalled above.  However, a question of major importance that has not received much attention so far is that of the size  of the estimated interval set. In effect, there exist in principle infinitely many interval estimators that satisfy the outer-bounding condition for the state trajectories of the system of interest.  But the best estimator would be the one that generates  the smallest possible interval sets  (in some sense) that contains the actual state. We therefore consider here the question of how to characterize the tightest interval-valued estimator for discrete-time linear systems. Note that this question was also tackled in  our previous works \cite{Bako18-CDC,Bako19-Automatica} but for the case of continuous-time LTI systems. The discrete-time case is sufficiently different to deserve a separate treatment. 
Note that the concept of tightness (of the interval-valued estimator) introduced here is different from the  notion of optimal design used in, for example, \cite{Dinh19-IJC,Wang18-SCL,Briat16-Automatica,Tang19-TAC}. The first characterizes the intersection of all intervals which contain the true state while the second refers to the selection of the parameters which optimize a certain performance function (e.g., peak-to-peak gain or $H_\infty$ norm) for a given parametrized structure of estimators. Indeed both approaches can be nicely combined as follows: select an optimal gain of the observer in the sense of some performance index an then construct the tightest estimator for this value of the gain using the approach of the current paper.   
Our approach leads to a tight estimator in convolutional form whose numerical implementation may, in the general case, be computationally costly. We therefore briefly discuss its realizability  by an LTI state-space form. Contrary to the continuous-time case, the realization problem  is numerically tractable in discrete-time. In case  a state-space realization does not exist, we consider some over-approximations of the tightest estimator.

The second part of the paper extends the discussion to switched linear systems. 
\noindent An increasing number of interval observer design methods are being proposed for this class of systems  (see e.g., \cite{Ethabet18-Automatica,Ifqir17-IFAC,Dinh19-IJC,Guo17-IFAC}). The main challenge that still needs to be overcome for most of these methods  is the conservatism of the design conditions. Here, we first discuss the expression of the tightest estimator for switched linear systems. We then consider a less computationally demanding relaxation of this tightest estimator  and propose a potentially less conservative design condition for the existence of a  stable interval-valued estimator under arbitrary switchings. Moreover, the search for the  estimator's parameters is numerically tractable and more precisely, expressible in the form of a convex feasibility problem.

\paragraph{Outline}The remainder of the paper is organized as follows. In Section \ref{sec:Preliminaries}, we set up the estimation problem and present the technical material to be used later for designing the estimator. 
In Section \ref{sec:Open-Loop} we discuss estimators in open-loop, that is, estimators that result only from the simulation of the state transition equation without any use of the measurement. Section \ref{sec:Closed-Loop} discusses a systematic way of transforming a classical observer into an interval-valued estimator in the new  estimation framework. Section \ref{sec:sls} considers an extension of the design method to switched linear systems.  Section \ref{sec:Simulations} reports some numerical results confirming tightness of the proposed estimator. We conclude the paper in Section \ref{sec:Conclusion}.  

\paragraph{Notations}
$\Re$ (resp. $\Re_+$) is the set of real (resp. nonnegative real) numbers; $\mathbb{Z}$ (resp. $\mathbb{Z}_+$) is the set of  (resp. nonnegative) integers. For  a real  number  $x$, $|x|$ will refer to the absolute value  of $x$.  For $x=\bbm x_1 & \cdots & x_n\eem^\top\in \Re^n$, $\left\|x\right\|_p$ will denote the $p$-norm of $x$ defined by $\left\|x\right\|_p=(|x_1|^p+\cdots+|x_n|^p)^{1/p}$, for $p\geq 1$. In particular for $p=\infty$, $\left\|x\right\|_\infty=\max_{i=1,\ldots,n}\left|x_i\right|$. For a matrix $A\in \Re^{n\times m}$, $\left\|A\right\|_F$ is the Frobenius norm of $A$ defined by  $\left\|A\right\|_F=\tr(A^\top A)^{1/2}$ (with $\tr$ referring to the trace of a matrix). 
 If $A=[a_{ij}]$ and $B=[b_{ij}]$ are real matrices of the same dimensions, the notation $A\leq B$ will be understood as an elementwise inequality on the entries, i.e., $a_{ij}\leq b_{ij}$ for all $(i,j)$.  $|A|$ corresponds to the matrix $[|a_{ij}|]$ obtained by taking the absolute value of each entry of $A$. For a square matrix $A$, $\rho(A)$ will refer to spectral radius of $A$ ; for a finite set $\Sigma=\left\{A_1,\ldots,A_m\right\}$ of square matrices, $\rho(\Sigma)$ will be the joint spectral radius of $\Sigma$ and $|\Sigma|$ will refer to the set $\left\{|A_1|,\ldots,|A_m|\right\}$.  
In case $A$ and $B$ are real square  matrices, $A\succeq B$ (resp. $A\succ B$) means that $A-B$ is positive semi-definite (resp. positive definite).

\section{Preliminaries}\label{sec:Preliminaries}
\subsection{Estimation problem settings}
Consider a Linear Time Invariant (LTI) system described by
\begin{equation}\label{eq:LTI}
	\begin{array}{rcl}
		x(t+1)&=&Ax(t)+Bw(t) \\
		y(t)&=& Cx(t)+v(t), 
	\end{array}
\end{equation}
where $x(t)\in \Re^{n}$, $w(t)\in \Re^{n_w}$, $y(t)\in \Re^{n_y}$, are respectively the state, control input and output at time $t\in \mathbb{Z}_+$;  $\left\{v(t)\right\}\subset\Re^{n_y}$ are unknown but bounded disturbances. $A\in \Re^{n\times n}$, $B\in \Re^{n\times n_w}$,  $C\in \Re^{n_y\times n}$ are some real matrices.  

Throughout the paper we make the following assumption concerning the external signals acting in system \eqref{eq:LTI}.  
\begin{assumption}\label{assum:Bounding}
There exist (known) bounded sequences $\left\{(\underline{w}(t), \overline{w}(t))\right\}$ and $\left\{(\underline{v}(t), \overline{v}(t))\right\}$ such that 
$\underline{w}(t)\leq w(t)\leq  \overline{w}(t) $ and $\underline{v}(t)\leq v(t)\leq  \overline{v}(t) $ for all \rev{$t\in\mathbb{Z}_+$.}
\end{assumption}

To begin with, let us fix some notation. Consider two vectors $\underline{x}$ and $\overline{x}$  in $\Re^n$ such that $\underline{x}\leq \overline{x}$ with the inequality holding componentwise.  An interval $\interval{\underline{x}}{\overline{x}}$ of $\Re^n$ is the subset defined by
\begin{equation}\label{eq:Interval}
	\interval{\underline{x}}{\overline{x}}= \big\{x\in \Re^n: \underline{x}\leq x \leq \overline{x}\big\}. 
\end{equation} 
We consider in this paper the problem of designing an \textit{interval-valued estimator} for the state of the LTI system \eqref{eq:LTI}. 
Considering that the initial state $x(0)$ of \eqref{eq:LTI} lives in an interval of the form $[\underline{x}(0), \overline{x}(0)]\subset\Re^{n}$ and that the external signals $w$ and $v$ are subject to Assumption \ref{assum:Bounding}, we want to estimate upper and lower bounds $\overline{x}(t)$ and $\underline{x}(t)$ for all the possible state trajectories of the uncertain system \eqref{eq:LTI}. 
\begin{definition}[Interval estimator]\label{def:Interval-Estimator}
Consider the system \eqref{eq:LTI} and pose $b_w(t)=\bbm \underline{w}(t)^\top & \overline{w}(t)^\top\eem^\top$, $b_v(t)=\bbm \underline{v}(t)^\top & \overline{v}(t)^\top\eem^\top$. Further, let $W^t = \big(b_w(0), \ldots,b_w(t)\big)$, $V^t=\big(b_v(0), \ldots,b_v(t)\big)$ and $Y^t=\big(y(0), \ldots,y(t)\big)$. Consider a dynamical system defined by 
\begin{equation}\label{eq:interval-estimator}
	\begin{aligned}
		&\underline{x}(t)=F_t(W^t,V^t,Y^t,X^0) \\
		&\overline{x}(t)=G_t(W^t,V^t,Y^t,X^0)
	\end{aligned} 
\end{equation}
where $F_t$ and $G_t$ are some functions indexed by time,  $(\underline{x}(t), \overline{x}(t))$ denote the output for any $t\in \mathbb{Z}_+$ and  $X^0=(\underline{x}(0), \overline{x}(0))$.  
The system \eqref{eq:interval-estimator} is called an \textit{interval-valued estimator} for system \eqref{eq:LTI} if: 
\begin{enumerate}
	\item Any state trajectory $x$ of \eqref{eq:LTI} satisfies   $\underline{x}(t)\leq x(t)\leq  \overline{x}(t) $  for all $t\in \mathbb{Z}_+$, whenever $\underline{x}(0)\leq x(0)\leq  \overline{x}(0) $
	\item \eqref{eq:interval-estimator} is Bounded Input-Bounded Output (BIBO) stable. 
\end{enumerate}
\end{definition}
\noindent Here the inputs of system \eqref{eq:interval-estimator} are the signals $b_w$, $b_v$, $y$ and the vector $X^0$. By BIBO stability we mean here that $\underline{x}$ and $\overline{x}$ in \eqref{eq:interval-estimator} are bounded whenever those input signals have bounded infinity norms. 

\begin{rem}
The estimator as expressed in \eqref{eq:interval-estimator} is in the general form of a causal dynamic system. As such it is not necessarily realizable in the form of a discrete-time state-space representation. This is why we will prefer throughout this paper the terminology \textit{ interval-valued estimator} to \textit{ interval observer} (the latter being conventionally used to refer to a state-space realization of the estimator, when it exists). 
\end{rem}

We will discuss two types of interval-valued estimators: open-loop interval estimators (or simulators) where \eqref{eq:interval-estimator} does not depend on the measurements $Y^t$ and the measurement noise $V^t$ ; and closed-loop interval estimators where measurement is fed back to the estimator.  
There are in principle infinitely many estimators that qualify as interval estimators in the sense of  Definition \ref{def:Interval-Estimator}. It is therefore desirable to define a performance index (measuring e.g. the size of the estimator) which selects the best estimator among all. We will be interested here in the smallest interval estimator in the following sense. 
\begin{definition}\label{def:smallest-interval}
Let $\mathcal{S}$ denote a subset of $\Re^n$. An interval $\mathcal{I}_{\mathcal{S}}\subset\Re^n$ is called the tightest interval containing $\mathcal{S}$ if $\mathcal{S}\subset \mathcal{I}_{\mathcal{S}}$ and if for any interval $\mathcal{J}$ of $\Re^n$, 
$\mathcal{S}\subset \mathcal{J} \: \Rightarrow \: \mathcal{I}_{\mathcal{S}} \subset \mathcal{J}.$
\end{definition}
\noindent In other words, the tightest interval $\mathcal{I}_{\mathcal{S}}$ "generated" by $\mathcal{S}$ is the intersection of all intervals containing $\mathcal{S}$. 
\subsection{Preliminary material on interval representation}\label{subsec:Interval-Representation} 
An important observation for future developments of the paper is that an interval $\interval{\underline{x}}{\overline{x}}$ of  $\Re^n$ can be equivalently represented by
\begin{equation}
	\mathcal{C}(c_x,p_x)=\big\{c_x+\diag\big(p_x\big) \alpha : \alpha\in \Re^n, \: \left\|\alpha\right\|_\infty \leq 1\big\}
\end{equation}
where 
 \begin{equation}
	 c_x=\dfrac{\overline{x}+\underline{x}}{2},  \quad p_x = \dfrac{\overline{x}-\underline{x}}{2}
 \end{equation}
The notation $\diag(v)$ for a vector $v\in \Re^n$ refers to  the diagonal matrix whose diagonal elements are the entries of $v$.   We will call the so-defined $c_x$ the center of the interval $\interval{\underline{x}}{\overline{x}}$ and $p_x$ its radius. To sum up, the interval set can be equivalently represented by the pairs $(\underline{x},\overline{x})\in \Re^n\times \Re^n$ and $(c_x,p_x)\in \Re^n\times \Re_+^n$  and $\interval{\underline{x}}{\overline{x}}=\mathcal{C}(c_x,p_x)$.   Finally, it will be useful to keep in mind that $\underline{x}=c_x-p_x$ and $\overline{x}=c_x+p_x$. 

The following lemma is a key result for our derivations. 
\begin{lem}\label{lem:Az+w}
Let $M\in \Re^{n\times m}$, $N\in \Re^{n\times n_f}$, $(\underline{f},\overline{f})\in \Re^{n_f}\times \Re^{n_f}$ and $(\underline{z},\overline{z})\in \Re^m\times \Re^m$ such that $\underline{f}\leq \overline{f}$ and $\underline{z}\leq \overline{z}$.  
Consider the set $\mathcal{I}$  defined by 
\begin{equation}\label{eq:x=Az+Bw}
	\mathcal{I} = \Big\{Mz+Nf: \underline{z}\leq z \leq \overline{z}, \: \underline{f}\leq f \leq \overline{f} \Big\}.
\end{equation}
Define the vectors $(c,p)$ by
 \begin{equation}\label{eq:C(c,p)}
	 \begin{aligned}
		 &c = Mc_z+Nc_f \\
		 &p = \left|M\right|p_z+\left|N\right|p_f, 
	\end{aligned}
\end{equation}
 with $p_z=(\overline{z}-\underline{z})/2$ and $p_f=(\overline{f}-\underline{f})/2$.\\
Then $\interval{c-p}{c+p}$  is the tightest interval containing $\mathcal{I}$ in the sense of Definition \ref{def:smallest-interval}. 
\end{lem}
\begin{proof}
The lemma is a simpler restatement of Lemma 1 in \cite{Bako19-Automatica}.
\end{proof}
\begin{rem}\label{rem:psi(A)}
If we let $\underline{x}=c-p$ and $\overline{x}=c+p$ with $(c,p)$ defined in \eqref{eq:C(c,p)}, then  $\interval{\underline{x}}{\overline{x}}=\mathcal{C}(c,p)$ and 
\begin{equation}\label{eq:BarX-BarZ}
	\BM\underline{x} \\ \overline{x}\EM=
	\Psi(M) \BM\underline{z} \\  \overline{z}\EM+\Psi(N)\BM \underline{f}\\ \overline{f}\EM
\end{equation}
with the notation $\Psi(\cdot)$ being defined by 
\begin{equation}\label{eq:PSI}
	\Psi(M)=\BM 
		\dfrac{M+\left|M\right|}{2} & \dfrac{M-\left|M\right|}{2}\\ 
		\dfrac{M-\left|M\right|}{2} & \dfrac{M+\left|M\right|}{2}
		\EM
\end{equation}
for any matrix $M$. If $M\in \Re^{n\times m}$, then $\Psi(M)$ is a matrix of dimensions $2n\times 2m$. 
It is particularly useful to note that the center-radius representation $(c,p)$ of $\interval{\underline{x}}{\overline{x}}$ follows from $(\underline{x},\overline{x})$ by a simple coordinates change as follows: 
\begin{equation}\label{eq:Similar-Tranform}
	\BM c\\ p\EM = T_n^{-1}\BM \underline{x} \\ \overline{x} \EM \: \: \mbox{ and } \: \:  
	T_n^{-1}\Psi(M)T_m= \BM M & 0\\ 0 & |M|\EM 
\end{equation}
where $T_n$ refers to the nonsingular matrix given by 
$$
	T_n=\BM I_n & -I_n\\ I_n & I_n\EM, 
$$
with $I_n$ denoting the identity matrix of order $n$. In the case where $M$ is a square matrix, then $n=m$ so that \eqref{eq:Similar-Tranform} implies that $\Psi(M)$ is similar to a block diagonal matrix.  
\end{rem}
We note that the first part of the statement of Lemma \ref{lem:Az+w} (i.e., the fact that $\mathcal{I}\subset \mathcal{C}(c,p)$)  also appeared in \cite{Efimov13-Automatica} and was proved using  a different line of arguments. The approach taken here is tailored to proving the second part, namely the fact that $\mathcal{C}(c,p)$ is indeed the tightest interval containing $\mathcal{I}$. This is  a key result in determining the tightest interval-valued state estimator.  
\subsection{Additional technical material}
For easier reference in this paper, we now state a couple of facts to be used in the stability proofs and in the bounding conditions. The next lemma recalls some useful basic facts from \cite[Chap. 8]{Horn85-Book}. 
\begin{lem}\label{lem:identities}
Let $A$ and $B$ be matrices of compatible dimensions. Then the following properties hold (componentwise): 
\begin{subequations}
\begin{align}
	&\left|A+B\right|\leq \left|A\right|+\left|B\right| \label{eq:sum}\\
	&\left|AB\right|\leq \left|A\right|\left|B\right|  \label{eq:product}\\
	&\left|A^r\right|\leq \left|A\right|^r \quad \forall \: r\in \mathbb{N} \label{eq:power}  \\
	&|A|\leq |B| \quad \Rightarrow \quad  \left\|A\right\|_F\leq \left\|B\right\|_F \label{eq:inequality-norm} \\
	&\left\|A\right\|_F=\left\|\left|A\right|\right\|_F \label{eq:equality-norm}
\end{align}
	\end{subequations}
The last equation just states that $A$ and $\left|A\right|$ have the same Frobenius norm. 
\end{lem}

\begin{definition}[Joint spectral radius of set of matrices \cite{Jungers09-Book}]
Let $\Sigma=\left\{A_1,\ldots,A_m\right\}$ be a finite set of square matrices in  $\Re^{n\times n}$.  Let $\left\|\cdot\right\|$ denote an arbitrary norm on $\Re^{n\times n}$. 
The \textit{joint spectral radius} of $\Sigma$ denoted $\rho(\Sigma)$ is defined as the number 
\begin{equation}\label{eq:RhoBar}
\rho(\Sigma)=\lim_{k\rightarrow\infty }\sup\left\{\|U_1\cdots U_k\|^{1/k}: U_j\in \Sigma \:  \forall j=1,\ldots,k \right\}. 	
\end{equation}
\end{definition}
\noindent In the particular case where the set $\Sigma$ contains a single matrix, the joint spectral radius reduces to the usual spectral radius of that single member of $\Sigma$.  As argued in \cite{Jungers09-Book}, the joint spectral radius of a compact set of matrices always exists and is independent of the norm involved in its definition \eqref{eq:RhoBar}.  For $\Sigma=\left\{A_1,\ldots,A_m\right\}$, we use the notation $|\Sigma|$ to denote the set of matrices $|\Sigma|=\left\{|A_1|,\ldots,|A_m|\right\}$ obtained by taking the absolute value of all matrices in $\Sigma$. We also define similarly $\Psi(\Sigma)=\left\{\Psi(A_1),\ldots,\Psi(A_m)\right\}$ with $\Psi(\cdot)$ defined as in \eqref{eq:PSI}.
\begin{lem}\label{lem:rho}
Let $\Sigma=\left\{A_1,\ldots,A_m\right\}$\ be a finite set of square matrices in $\Re^{n\times n}$.  Then 
\begin{subequations}
\begin{equation}\label{eq:RhoSigma}
	\rho(\Sigma)\leq \rho(|\Sigma|)
\end{equation}
\begin{equation}\label{eq:Rho}
\rho(|\Sigma|)=\rho(\Psi(\Sigma))
\end{equation}
\end{subequations} 
\end{lem}
\begin{proof}
The inequality \eqref{eq:RhoSigma} follows directly from the identities \eqref{eq:product}, \eqref{eq:inequality-norm} and \eqref{eq:equality-norm} stated above. In effect, since, as remarked above, the joint spectral radius is independent of the norm used in its definition \eqref{eq:RhoBar}, let us employ the Frobenius norm. 
Then $$ \begin{aligned}
	\rho(\Sigma)&=\lim_{k\rightarrow\infty }\sup\left\{\|U_1\cdots U_k\|_F^{1/k}: U_j\in \Sigma \:  \forall j=1,\ldots,k \right\}\\
	&\leq \lim_{k\rightarrow\infty }\sup\left\{\||U_1|\cdots |U_k|\|_F^{1/k}: U_j\in \Sigma \:  \forall j=1,\ldots,k \right\}\\
	& = \rho(|\Sigma|)
\end{aligned}$$ 
Note that the inequality above is a consequence of \eqref{eq:product}, \eqref{eq:inequality-norm} and \eqref{eq:equality-norm} by which $\|U_1\cdots U_k\|_F=\||U_1\cdots U_k|\|_F\leq \||U_1|\cdots |U_k|\|_F$ since $|U_1\cdots U_k|\leq |U_1|\cdots |U_k|$ by \eqref{eq:product}.  \\
In virtue of the relation \eqref{eq:Similar-Tranform}, to obtain \eqref{eq:Rho}, it suffices to use the invariance property of the joint spectral radius under similarity transformation \cite{Jungers09-Book}. In effect, according to \eqref{eq:Similar-Tranform} all matrices contained in $\Psi(\Sigma)$ are commonly reducible to block-diagonal matrices in  the sense that there exists a matrix $T_n$ as defined after Eq. \eqref{eq:Similar-Tranform}  such that for all $A_i\in \Sigma\subset \Re^{n\times n}$, 
$$T_n^{-1}\Psi(A_i)T_n= \BM A_i & 0\\ 0 & |A_i|\EM $$  
By then applying Proposition 1.5 in \cite{Jungers09-Book} and the inequality \eqref{eq:RhoSigma}, we have 
$\rho(\Psi(\Sigma))=\max\left(\rho(\Sigma),\rho(|\Sigma|)\right)\allowbreak=\rho(|\Sigma|)$.  
\end{proof}
In the special case where $\Sigma$ is reduced to a singleton, \eqref{eq:RhoSigma} recovers the more standard relation 
$\rho(A)\leq \rho(|A|)$. 

\begin{lem}\label{lem:stability-implication}
Consider two finite collections of nonnegative matrices $\Sigma=\left\{A_i\in \Re_+^{n\times n}:  i\in \mathbb{S}\right\}$ and $\overline{\Sigma}=\left\{\bar{A}_i\in \Re_+^{n\times n}:   i\in \overline{\mathbb{S}}\right\}$, where $\mathbb{S}$ and $\overline{\mathbb{S}}$ are finite sets with possibly different cardinalities. 
If for any $i\in \mathbb{S}$, there is $j\in \overline{\mathbb{S}}$ such that $A_i\leq \bar{A}_{j}$, then 
$\rho(\Sigma)\leq  \rho(\overline{\Sigma}).$
\end{lem}
\begin{proof}
Clearly, it follows from the assumption of the lemma that for any $(i_1,\ldots,i_q)\in \mathbb{S}^q$, there is $(j_1,\ldots,j_q)\in \overline{\mathbb{S}}^q$ such that $0\leq A_{i_1}\cdots A_{i_q}\leq \bar{A}_{j_1}\cdots \bar{A}_{j_q}$. 
The result then follows by applying the statement \eqref{eq:inequality-norm} of Lemma \ref{lem:identities}.  
\end{proof}


\section{Open-loop state interval estimator for LTI systems}\label{sec:Open-Loop}

\subsection{Open-loop simulation: the best interval estimator}
We first discuss a  simulation (that is, an estimation without using the measurement) of the state trajectory of  \eqref{eq:LTI} under an uncertain input sequence $\left\{w(t)\right\}$  and when the initial state $x(0)$ is drawn from a known  interval set. 
For this purpose, we assume that the  matrices $A$ and $B$ have fixed and known values and further, that $A$ is a Schur matrix. Then designing the tightest estimator boils down to  searching for the smallest sequence (in the sense of Definition \ref{def:smallest-interval}) $\left\{[\underline{x}(t),\overline{x}(t)]:t\in \mathbb{Z}_+\right\}$ of interval sets of $\Re^n$ which bound all the possible state trajectories generated by $\left\{w(t)\right\}$ and the uncertain initial state through the first equation of \eqref{eq:LTI}.

Let $(c_x(0),p_x(0))$ and $(c_w(t),p_w(t))$, $t\in  \mathbb{Z}_+$, denote the center-radius representations for the interval-valued initial state $[\underline{x}(0),\overline{x}(0)]$ and uncertain input  $[\underline{w}(t),\overline{w}(t)]$  respectively. Then the next theorem characterizes the tightest interval-valued state estimate in open-loop for system  \eqref{eq:LTI}. 
%
\begin{thm}\label{thm:Tight-Open-Loop}
Consider system \eqref{eq:LTI} under the assumption that $\rho(A)<1$. 
Then the intervals  $[\underline{x}(t),\overline{x}(t)]$, $t\in \mathbb{Z}_+$,   
$$ \underline{x}(t)=c_x(t)-p_x(t) \: \mbox{ and }\:  \overline{x}(t)=c_x(t)+p_x(t),$$
where 
\begin{align}
&c_x(t)  =A^tc_x(0)+\sum_{k=0}^{t-1}A^{t-1-k}B c_w(k) \label{eq:cx(t)}\\
&p_x(t)=\left|A^t\right|p_x(0)+\sum_{i=0}^{t-1}\left|A^{t-1-i}B\right|p_w(i), \label{eq:px(t)}
\end{align} 
 define  the tightest (open-loop) interval-valued estimator for system \eqref{eq:LTI} in the sense of Definition \ref{def:smallest-interval}. 
\end{thm} 
\begin{proof}
The proof is an immediate consequence of Lemma \ref{lem:Az+w}. In effect, 
applying repeatedly the first equation in \eqref{eq:LTI} yields
\begin{equation}\label{eq:simulated-state}
	\begin{aligned}
		x(t) & =A^tx(0)+\sum_{k=0}^{t-1}A^{t-1-k}B w(k)\\
		  & = A^tx(0)+\BBM A^{t-1}B & \cdots & AB & B\EEM w_{0:t-1}
	\end{aligned}
\end{equation}
with the notation $w_{0:t}$ defined by $w_{0:t}=\bbm w(0)^\top & \cdots & w(t)^\top\eem^\top$. 
Then by applying Lemma \ref{lem:Az+w} with $z$ and $f$ replaced by $x(0)$ and $w_{0:t-1}$, and $M$ and $N$ replaced by $A^t$ and  $\BBM A^{t-1}B & \cdots & AB & B\EEM$ respectively, we can conclude that the interval sequence defined by \eqref{eq:cx(t)}-\eqref{eq:px(t)}  is a bounding one for the state of system \eqref{eq:LTI} and is also the tightest. As to the BIBO stability, it is also immediate since $A$ is Schur stable. 
\end{proof}
Applying directly Eqs \eqref{eq:cx(t)} and \eqref{eq:px(t)} at each time step to compute $c_x$ and $p_x$  may be very costly (and even unaffordable) when the time horizon for estimation gets  large. It is therefore  desirable to search instead for a one-step ahead difference equations for implementing the derived tightest estimator. In this respect, we can remark that $c_x$ in \eqref{eq:cx(t)} can be easily realized by a state-space representation of the form 
\begin{equation}\label{eq:c_x-state-space}
	c_x(t+1)=Ac_x(t)+Bc_w(t).
\end{equation}
However, realizing $p_x$ in \eqref{eq:px(t)} is quite challenging in general. Though in the specific situations where the entries of $A$ and $B$ have the same sign, $p_x$ satisfies $p_x(t+1)=\left|A\right|p_x(t)+\left|B\right|p_w(t)$. 

\subsection{Exact state-space realization of the estimator }\label{subsec:realization}
In a more general situation one can search for \textit{linear time-invariant realization} of system \eqref{eq:px(t)} (whose input and output are respectively $p_w$ and $p_x$) independently of the class of inputs. The question is that of finding a set of matrices $\left(\mathcal{A},\mathcal{B},\mathcal{C},\phi_0\right)\in \Re^{d\times d}\times \Re^{d\times n_w}\times \Re^{n\times d}\times \Re^{d}$ for some dimension $d$ and such that the solution $\phi$ of the difference equation ${\phi}(t+1)=\mathcal{A}\phi(t)+\mathcal{B}p_w(t)$, $\phi(0)=\phi_0$,  satisfies $p_x(t) = \mathcal{C}\phi(t)$ for all $t\in \mathbb{Z}_+$. Indeed this is true if and  only if
$\mathcal{C}\mathcal{A}^t\tilde{\mathcal{B}}=H(t)$
where $\tilde{\mathcal{B}}=\bbm \mathcal{B} & \phi_0\eem$ and $H$ is the impulse response of the system \eqref{eq:px(t)} defined by $H(t)\triangleq\bbm|A^tB| & |A^t|p_x(0)\eem$. 
This can be expressed in term of a rank condition. For any positive integers $(i,j)$ consider the block Hankel matrix defined by 
\begin{equation}\label{eq:Hankel}
	\mathcal{H}_{i,j}=\left[\begin{array}{cccccc} H(0) & H(1) & &\cdots& &H(j-1)\\ 
											H(1) & H(2) & &\cdots& & H(j) \\
											\vdots & \vdots & &\vdots& & \vdots\\
											H(i-1) & H(i) & &\cdots& & H(i+j-2)\end{array}\right]
	\end{equation}
Then a necessary and sufficient condition for the existence of an LTI realization of system \eqref{eq:px(t)} is obtained as follows  \cite[p.125]{Casti12-Book}. 
\begin{thm}\label{thm:realization}
System \eqref{eq:px(t)} is LTI realizable in finite dimension if and only if there exist integers $r$, $l$ and $m$ such that
\begin{equation}\label{eq:realizability}
	\rank\left(\mathcal{H}_{r,l}\right)=\rank\left(\mathcal{H}_{r+1,l+j}\right)=m<\infty \quad \forall j\geq 1. 
\end{equation}
\end{thm}
Standard realization algorithms can be employed to compute, whenever the condition of Theorem \ref{thm:realization} is satisfied,  a minimal LTI realization $(\mathcal{A},\tilde{\mathcal{B}},\mathcal{C})$ of \eqref{eq:px(t)}. 
For more on this matter, the interested reader is referred to, e.g., \cite{Casti12-Book,DeSchutter00,VanDenHof97-LAA}.


\subsection{Over-approximations}\label{subsec:Approximation}
When the realizability condition \eqref{eq:realizability} fails to hold, then obtaining the tightest interval estimator requires computing the smallest radius  from its convolutional expression given in \eqref{eq:px(t)}.  However this might become quickly impractical as time goes to infinity.   So, in practice, by  default of being able to realize $p_x$ with a finite-dimensional state-space representation, a relaxation may be necessary. 
Here we discuss two methods for overestimating $p_x(t)$. 

\subsubsection{A truncated approximation}\label{subsub:Truncated}
The first approximation method relies on an expansion of $x(t)$ over a sliding time horizon of fixed size. More specifically, by noting that $x(t) = A^{q}x(t-q)+\sum_{k=t-q}^{t-1} A^{t-1-k}Bw(k)$
for some fixed $q$, we can invoke Lemma \ref{lem:Az+w} to obtain an overestimate of $p_x$ as given in the following proposition. 
\begin{prop}\label{prop:pxq}
Let $q$ be a fixed positive integer. Consider the signal $\hat{p}_{x,q}:\mathbb{Z}_+\rightarrow\Re_+^{n}$ defined by 
\begin{equation}\label{eq:phat}
	\hat{p}_{x,q}(t)=
\left\{	\begin{array}{ll}
	\left|A^t\right|{p}_x(0)+ \sum_{k=0}^{t-1} |A^{t-1-k}B|p_w(k) &  \mbox{ if } t=0,\ldots,q\\
	\left|A^q\right|\hat{p}_{x,q}(t-q)+ \sum_{k=t-q}^{t-1} |A^{t-1-k}B|p_w(k), & \mbox{ if } t> q
	 \end{array}\right.
\end{equation}
with $A$ and $B$ being the matrices of the system \eqref{eq:LTI} and $p_x(0)$ and $p_w$ defined as in the statement of Theorem \ref{thm:Tight-Open-Loop}.\\
Then the following statements are true: 
\begin{enumerate}
	\item ${p}_{x}(t)$ is upper-bounded componentwise by $\hat{p}_{x,q}(t)$ 
	\begin{equation}\label{eq:Ineq-px}
	p_x(t)\leq \hat{p}_{x,q}(t) \quad \forall t\in \mathbb{Z}_+
\end{equation}
 so that
$$\interval{c_x(t)-p_{x}(t)}{c_x(t)+p_x(t)}\subset \interval{c_x(t)-\hat{p}_{x,q}(t)}{c_x(t)+\hat{p}_{x,q}(t)} \quad \forall t\in \mathbb{Z}_+$$
where $c_x$ and $p_x$ are defined in \eqref{eq:cx(t)}-\eqref{eq:px(t)}. 
\item If in addition, $\rho(A)<1$, then there exists an integer $q^\star\in \mathbb{Z}_+$ such that the sequence $\interval{c_x(t)-\hat{p}_{x,q}(t)}{c_x(t)+\hat{p}_{x,q}(t)}$, $t\in \mathbb{Z}_+$,  defines an interval-valued estimator for all  $q\geq q^\star$. 
\end{enumerate}
\end{prop}
\begin{proof}
To establish \eqref{eq:Ineq-px}, we start by observing that $\hat{p}_{x,q}(t)=p_x(t)$ for all $t=0,\ldots,q$. Hence, \eqref{eq:Ineq-px} is satisfied for $t=0,\ldots,q$. Now if $t>q$, it is possible to write it in the form $t=q\alpha(t)+r(t)$ for some positive integers $(\alpha(t),r(t))$ such that $0\leq r(t)<q$. 
On the other hand, we have $\hat{p}_{x,q}(t)=\left|A^q\right|\hat{p}_{x,q}(t-q)+ \sum_{k=t-q}^{t-1} |A^{t-1-k}B|p_w(k)$ in this case. Iterating this equation for $\hat{p}_{x,q}(t-q)$, $\hat{p}_{x,q}(t-2q)$, etc, ultimately yields
$$\hat{p}_{x,q}(t)=|A^q|^{\alpha(t)}\hat{p}_{x,q}(r(t))+\sum_{j=1}^{\alpha(t)}\: \sum_{k=t-jq}^{t-1-(j-1)q}|A^{t-1-k}B|p_w(k) $$ 
Recall that $\hat{p}_{x,q}(r(t))=p_x(r(t))=\left|A^{r(t)}\right|{p}_x(0)+ \sum_{k=0}^{r(t)-1} |A^{r(t)-1-k}B|p_w(k)$. Replacing in the expression of $\hat{p}_{x,q}(t)$ and invoking  the identities \eqref{eq:product}-\eqref{eq:power} stated in Lemma \ref{lem:identities} gives
$$\begin{aligned}
	\hat{p}_{x,q}(t)&=|A^q|^{\alpha(t)}|A^{r(t)}|{p}_x(0)+\sum_{k=0}^{r(t)-1} |A^q|^{\alpha(t)} |A^{r(t)-1-k}B|p_w(k)+
		\sum_{k=r(t)}^{t-1}|A^{t-1-k}B|p_w(k)\\ 
	&\geq |A^{q\alpha(t)+r(t)}|{p}_x(0)+\sum_{k=0}^{r(t)-1} |A^{q\alpha(t)+r(t)-1-k}B|p_w(k)+\sum_{k=r(t)}^{t-1}|A^{t-1-k}B|p_w(k)\\
	& = |A^{t}|{p}_x(0)+\sum_{k=0}^{t-1} |A^{t-1-k}B|p_w(k)\\
	&= p_x(t).
\end{aligned} 
$$ 
We have hence established that $p_x(t)\leq \hat{p}_{x,q}(t)$ for all $t\in \mathbb{Z}_+$. \\
To prove the second statement, it suffices to show that there exists an integer $q^\star$ such that  $\rho(|A^q|)<1$ $\forall q\geq q^\star$. To see this, recall that $\rho(A)<1$ implies that $\lim_{q\rightarrow +\infty}\left\|A^q\right\|_F=0$. Hence,  there exists an integer $q^\star$ such that $\left\|A^q\right\|_F<1$ for all $q\geq q^\star$. On the other hand we have\footnote{It is known from \cite[Prop. 1.4]{Jungers09-Book} that 
the spectral radius $\rho$ of a matrix $A\in \Re^{n\times n}$ obeys $\rho(A)=\inf_{\left\|\cdot\right\|}\left\|A\right\|$ with the infimum being taken over all norms on $\Re^{n\times n}$.} $\rho(|A^q|)\leq \left\||A^q|\right\|_F=\left\|A^q\right\|_F$.  
\end{proof}

\noindent What  Proposition \ref{prop:pxq} says is that if we choose appropriately the integer $q$, then the sequence of intervals $\interval{c_x(t)-\hat{p}_{x,q}(t)}{c_x(t)+\hat{p}_{x,q}(t)}$, $t\in \mathbb{Z}_+$, with $\hat{p}_{x,q}(t)$  expressed as in \eqref{eq:phat},  defines an interval estimator for system \eqref{eq:LTI} in the sense of Definition \ref{def:Interval-Estimator}. This interval estimator always contains the tightest one characterized in Theorem \ref{thm:Tight-Open-Loop}.  Naturally, the larger $q$, the higher the computational cost associated with the implementation but the smaller the radius $\hat{p}_{x,q}(t)$ and the error $\hat{p}_{x,q}(t)-p_x(t)$. 

\begin{rem}\label{rem:Approx-order-one}
If $\rho(\left|A\right|)<1$, then  taking $q=1$ in \eqref{eq:phat} yields a much simpler interval estimator $(c_x(t),\hat{p}_{x,1}(t))$ with $c_x$ defined  in \eqref{eq:c_x-state-space} and $\hat{p}_{x,1}$ defined as in \eqref{eq:phat} admitting directly a state-space representation,
\begin{equation}\label{eq:phat1}
	\hat{p}_{x,1}(t+1)=\left|A\right|\hat{p}_{x,1}(t)+\left|B\right|p_w(t).
\end{equation}
In contrast, the resulting radius $\hat{p}_{x,1}(t)$ is larger, that is, the associated bounds of the interval-valued estimate are looser. 
One can further observe that the interval estimator \eqref{eq:phat1} can indeed be  written in the more classical observer form (see Remark \ref{rem:psi(A)}):
$$	\BM\underline{x}(t+1) \\ \overline{x}(t+1)\EM=
	\Psi(A) \BM\underline{x}(t) \\  \overline{x}(t)\EM+\Psi(B)\BM \underline{w}(t)\\ \overline{w}(t)\EM $$
	where $\Psi$ is defined as in \eqref{eq:PSI}. 
\end{rem}

\subsubsection{Over-approximating the input}
The second approximation method makes use of the following proposition (whose  proof follows by simple calculations). 
\begin{prop}\label{eq:px-for-constant-pw}
Assume that $p_w(t)=p_w(0)$ for all $t\in \mathbb{Z}_+$, i.e., $p_w$ is constant. Then the sequence $\left\{p_x(t)\right\}$ in \eqref{eq:px(t)} can be realized as follows:
\begin{equation}\label{eq:realization-px}
	\left\{\begin{aligned}
		&M(t+1)=AM(t), \quad M(0)=I_{n}\\
		&r(t+1) = r(t)+\left|M(t)B\right|p_w(0), \quad r(0)=0\\
		&p_x(t)=\left|M(t)\right|p_x(0)+r(t)
	\end{aligned}\right.
\end{equation}
with state $(M(t),r(t))\in \Re^{n\times n}\times \Re^{n}$ and $I_{n}$ being the identity matrix of order $n$. 
\end{prop}
Indeed by Assumption \ref{assum:Bounding}, $\left\{p_w(t)\right\}$ is bounded. Therefore, let $r^o$ be the vector in $\Re^{n_w}$ whose $i$-th   entry denoted $r_i^o$ is defined by $r_i^o=\max_{t\in \mathbb{Z}_+}p_{w,i}(t)$ where $p_{w,i}(t)$ refers to the $i$-th entry of $p_w(t)$. Then $\left\{w(t)\right\}$ satisfies $c_w(t)-r^o\leq w(t)\leq c_w(t)+r^o$ and hence $(c_w(t),r^o)$ is a valid but looser interval representation for the input signal $w$ which fulfills the condition of Proposition \ref{eq:px-for-constant-pw}.  As a consequence, replacing $p_w(0)$ in \eqref{eq:realization-px} with $r^o$ gives a computable realization of an interval estimator for the state of system \eqref{eq:LTI}. 

\section{Closed-loop state estimator}\label{sec:Closed-Loop}

\subsection{A formal characterization of the best linear estimator}
In case the system \eqref{eq:LTI} is not stable but detectable\footnote{i.e., all its unobservable modes are stable.}, it is possible to find a matrix gain $L$ such that $A-LC$ is Schur stable. We can then construct an interval-valued state estimator from the classical observer form. 
For example, departing from the structure of the classical Luenberger observer, it is easy to see that the true state of system  of \eqref{eq:LTI} satisfies 
\begin{equation}\label{eq:observer}
	x(t+1) =F(L)x(t)+G(L)s(t)
\end{equation}
where $F(L)=A-LC$, $G(L)=\bbm B & L & -L\eem$ and $s(t) = \bbm  w(t)^\top  & y(t)^\top & v(t)^\top\eem^\top$. Since this equation has the same form as \eqref{eq:LTI} and $L$ is assumed be such that $A-LC$ is Schur stable, all the preceding discussion in Section \ref{sec:Open-Loop} is applicable to the closed-loop case.
Note that beyond the possibility of retrieving stability from linear output injection as in \eqref{eq:observer}, it is intuitive that exploiting the measurements can help tighten further the interval-valued estimate of the state.   

\noindent For a given $L\in \mathcal{L}\triangleq\big\{L\in \Re^{n\times n_y}:\rho(A-LC)<1\big\}$, define $\mathcal{E}_L$ to be  the set of all pairs of vector-valued signals $(\underline{x},\overline{x})$ on $\mathbb{Z}_+$ which (lower and upper) bound the state $x$  under the constraint \eqref{eq:observer}, 
$$\mathcal{E}_L=\left\{(\underline{x},\overline{x}):\begin{array}{l}
	  \underline{x}(t)\leq x(t)\leq \overline{x}(t) \: \forall t\in \mathbb{Z}_+ \\
	 x(t+1)=F(L)x(t)+G(L)s(t) \:  \forall t\in \mathbb{Z}_+ \\ 
	  \underline{s}(t)\leq s(t)\leq \overline{s}(t) \: \forall t\in \mathbb{Z}_+\\
	  \underline{x}(0)\leq x(0)\leq \overline{x}(0) 
\end{array} 
\right\}$$
We give below a formal characterization of the best estimator over all members of \rev{$\mathcal{E}\triangleq \cup_{L\in \mathcal{L}}\mathcal{E}_L$.} 
\begin{thm}\label{thm:Tightest-CL}
Consider the LTI system \eqref{eq:LTI} under Assumption \ref{assum:Bounding}. 
Then the sequence of intervals  $[\underline{x}^\star(t),\overline{x}^\star(t)]$ defined by 
\begin{align}
	&\underline{x}_i^\star(t)=\sup_{L\in \mathcal{L}}\left[\hat{c}_{i}(t,L)-\hat{p}_{i}(t,L)\right] \label{eq:xbar-lower-star} \\
	&\overline{x}_i^\star(t)=\inf_{L\in \mathcal{L}}\left[\hat{c}_{i}(t,L)+\hat{p}_{i}(t,L)\right], \label{eq:xbar-upper-star}
\end{align}
for all $t\in \mathbb{Z}_+$ and all $i=1,\ldots,n$, with 
\begin{align}
&\hat{c}(t,L)  =F(L)^tc_x(0)+\sum_{k=0}^{t-1}F(L)^{t-1-k}G(L) c_s(k) \label{eq:cx(t)-CL}\\
&\hat{p}(t,L)\!=\!\left|F(L)^t\right|p_x(0)\!+\!\sum_{k=0}^{t-1}\left|F(L)^{t-1-k}G(L)\right|p_s(k), \label{eq:px(t)-CL}
\end{align} 
 \noindent define  the tightest interval-valued state estimator for system \eqref{eq:LTI}  within the class $\mathcal{E}$  in the sense of Definition \ref{def:smallest-interval}. 

\noindent Here, notation of the type $\hat{c}_{i}(t,L)$ refers to the $i$-th entry of the vector $\hat{c}(t,L)\in \Re^n$. 
\end{thm}
\begin{proof}
For a fixed $L\in \mathcal{L}$ we know from Theorem   \ref{thm:Tight-Open-Loop} that $(\hat{c}(\cdot,L),\hat{p}(\cdot,L))$ in \eqref{eq:cx(t)-CL}-\eqref{eq:px(t)-CL} is the center-radius representation of the tightest interval estimator for system \eqref{eq:LTI} under the constraint \eqref{eq:observer}. If we now let $L$ vary over $\mathcal{L}$, then naturally the pair $(\underline{x}^\star,\overline{x}^\star)$ as defined in \eqref{eq:xbar-lower-star}-\eqref{eq:xbar-upper-star} gives the tightest interval-valued state estimator for system \eqref{eq:LTI} of the class $\mathcal{E}$. 
\end{proof}

\noindent We note that $(\underline{x}^\star,\overline{x}^\star)$ does not necessarily lie in $\mathcal{E}$.  
It is also worth pointing out that from a computational perspective, the estimator $(\underline{x}^\star,\overline{x}^\star)$ defined by  \eqref{eq:xbar-lower-star}-\eqref{eq:xbar-upper-star} is very hard to handle. We will therefore study further the type of approximation given in Proposition \ref{prop:pxq} for a fixed value of $L\in \mathcal{L}$.

\subsection{A systematic design method}
In virtue of Lemma \ref{lem:Az+w}, an interval estimator can, in principle, be obtained from any standard linear observer. 
A general recipe for constructing such an interval estimator consists in (i) adjusting the equation of the classical observer so that it is satisfied by the true state as in  \eqref{eq:observer}; and  then  (ii)  applying Lemma \ref{lem:Az+w} either to \eqref{eq:observer} directly or  to a truncated expansion of it similar to the one discussed in  the beginning of Section \ref{subsub:Truncated}. Following this logic, let us apply the truncated interval estimator stated in Proposition \ref{prop:pxq} to the system \eqref{eq:observer} (which, recall, is a standard observer-based expression of the state of system \eqref{eq:LTI}). Denote with $(c_{x}^{\cl}, {p}_{x,q}^{\cl})$ the center-radius representation of the resulting  interval-valued estimator. Then for a given $L\in \mathcal{L}$, we have $c_{x}^{\cl}(t)=\hat{c}(t,L)$ $\forall t\in \mathbb{Z}_+$, with $\hat{c}(t,L)$ defined in \eqref{eq:cx(t)-CL}. As already mentioned, such a vector-valued signal $c_{x}^{\cl}$ is realizable in state-space form by  the following difference equation
\begin{equation}\label{eq:cx-CL}
	c_x^{\cl}(t+1)=F(L)c_x^{\cl}(t)+G(L)c_s(t), \quad   c_x^{\cl}(0)=c_x(0). 
\end{equation}


\noindent As to the radius ${p}_{x,q}^{\cl}:\mathbb{Z}_+\rightarrow \Re_{+}^n$, it is defined by 
\begin{equation}\label{eq:phat-CL}
	{p}_{x,q}^{\cl}(t)=
\left\{	\begin{array}{ll}
	\left|F(L)^t\right|{p}_x(0)+ \sum_{k=0}^{t-1} |F(L)^{t-1-k}G(L)|p_s(k) &  \mbox{ if } t=0,\ldots,q\\
	\left|F(L)^q\right|{p}_{x,q}^{\cl}(t-q)+ \sum_{k=t-q}^{t-1} |F(L)^{t-1-k}G(L)|p_s(k), & \mbox{ if } t> q
	 \end{array}\right.
\end{equation}
Here, $(c_s(t), p_s(t))\in \Re^{n_s}\times\Re_+^{n_s}$, $n_s=n_w+2n_y$, is a center-radius representation of $s(t)$ for $t\in \mathbb{Z}_+$. In explicit expressions, we have $c_s(t)=\bbm  c_w(t)^\top  & y(t)^\top & c_v(t)^\top\eem^\top$ and 
 $p_s(t)=[\begin{matrix}  p_w(t)^\top  & \mathbf{0}_{n_y}^\top & p_v(t)^\top\end{matrix}]^\top$ 
with $\mathbf{0}_{n_y}$ denoting a $n_y$-dimensional vector containing zeros.  

\noindent In the special case where $q=1$, the pair $p_{x,q}^{\cl}$ can be very simply realized as  
\begin{equation}
	\label{eq:Interval-Luenberger}
	\begin{aligned}
	&{p}_{x,1}^{\cl}(t+1)=\left|F(L)\right|{p}_{x,1}^{\cl}(t)+\left|G(L)\right|p_s(t), \quad   p_{x,1}^{\cl}(0)=p_x(0).
	\end{aligned}
\end{equation}
\noindent\rev{In this latter particular case, by expressing $\underline{x}(t)=c_x^{\cl}(t)-p_{x,1}^{\cl}(t)$ and $\overline{x}(t)=c_x^{\cl}(t)+p_{x,1}^{\cl}(t)$ from the solution $(c_x^{\cl},p_{x,1}^{\cl})$ of the system  given in \eqref{eq:cx-CL} and \eqref{eq:Interval-Luenberger} respectively, we obtain a dynamical system which is similar to the interval observer given in \cite{Efimov13-ECC}.
Note however that both are not the same  as the current expression provides a more flexible way for handling  the measurement noise. Further, as will be stated below, our condition of stability (which is $\rho(|A-LC|)<1$) reads very simply in this case.}

The following statement can be obtained. 
\begin{thm}\label{thm:Gain-L-stability}
Consider the system \eqref{eq:LTI} under Assumption \ref{assum:Bounding}. If the pair $(A,C)$ is detectable, then there exists a matrix gain $L\in \Re^{n\times n_y}$ and an integer $q^\star$ such that for any $q\geq q^\star$, $(c_{x}^{\cl}, {p}_{x,q}^{\cl})$ expressed in \eqref{eq:cx-CL}-\eqref{eq:phat-CL} defines an interval-valued estimator for system \eqref{eq:LTI}. 
In particular, if  $\rho(\left|A-LC\right|)<1$, then  $(c_x^{\cl},p_{x,1}^{\cl})$ given in \eqref{eq:cx-CL} and \eqref{eq:Interval-Luenberger} respectively, defines an interval-valued estimator  for system \eqref{eq:LTI}.  
\end{thm}
\begin{proof}
We already know from the first statement of Proposition \ref{prop:pxq}  that  $(c_x^{\cl},{p}_{x,q}^{\cl})$ satisfies the bounding condition of Definition  \ref{def:Interval-Estimator} in the sense that the state of system \eqref{eq:LTI} obeys $c_x^{\cl}(t)-p_{x,q}^{\cl}(t)\leq x(t)\leq c_x^{\cl}(t)+p_{x,q}^{\cl}(t)$. It remains to prove the BIBO stability of the system $(p_x(0),p_s)\mapsto (c_x^{\cl},{p}_{x,q}^{\cl})$.  
For this purpose, recall from control theory that the detectability property of $(A,C)$ guarantees the existence of  $L\in \Re^{n\times n_y}$ such that $\rho(A-LC)<1$. We can then find, as shown in the proof of Proposition \ref{prop:pxq}, point 2), an integer $q^\star$ such that $\rho(|A-LC|^q)<1$ for any $q\geq q^\star$. Hence for any $L$ satisfying $\rho(A-LC)<1$   there exists $q^\star$ such that for any $q\geq q^\star$, $(c_{x}^{\cl}, {p}_{x,q}^{\cl})$ constitutes an interval-valued estimator for system \eqref{eq:LTI} in the sense of Definition  \ref{def:Interval-Estimator}. In the particular case where $q$ is taken equal to $1$,  the BIBO stability condition for $(p_x(0),p_s)\mapsto (c_x^{\cl},{p}_{x,1}^{\cl})$ holds if  $\rho(\left|A-LC\right|)<1$ hence proving the last statement of the theorem.     
\end{proof}

\noindent To obtain the estimator $(c_{x}^{\cl}, {p}_{x,q}^{\cl})$ we just need to select a matrix $L\in \Re^{n\times n_y}$ obeying $\rho(A-LC)<1$ and then find the minimum $q^\star$ by testing incrementally whether $\rho(|A-LC|^q)<1$ for $q=1,2, \ldots$ until obtaining a positive return. From a numerical standpoint, obtaining $L$ and $q^\star$ under the stability constraint is quite easy by standard routines. Observability (or just detectability) of $(A,C)$ is a sufficient condition for ensuring the existence of an interval-valued estimator of the form $(c_{x}^{\cl}, {p}_{x,q}^{\cl})$.  Now if we impose $q=1$ for the design (for example, for computational reasons), then $L$ must satisfy $\rho(\left|A-LC\right|)<1$ as prescribed by Theorem \ref{thm:Gain-L-stability}.  In this latter case, a question is how to effectively select a matrix gain $L\in \Re^{n\times n_y}$, when possible, so as to realize this condition. An answer is provided by the following lemma. A similar (but  only sufficient) result was obtained in \cite{Efimov13-ECC}.  
\color{black}

\begin{lem}\label{lem:LMI}
The following statements are equivalent:
\begin{enumerate}
	\item[(a)] There exists $L\in \Re^{n\times n_y}$  such that $\rho(\left|A-LC\right|)<1$.   
\item[(b)] There exist a \textit{diagonal positive definite} matrix $P\in \Re_+^{n\times n}$  and some matrices $Y\in \Re^{n\times n_y}$, $X\in \Re_+^{n\times n}$ satisfying the conditions:
\begin{equation}\label{eq:LMI}
	\begin{aligned}
		&\BM P & X\\ X^\top & P \EM\succ 0\\
		& |P A-YC|\leq X\\
	\end{aligned}
\end{equation}
\end{enumerate}
Whenever \eqref{eq:LMI} holds, the matrix $L$ in (a) is given by $L=P^{-1}Y$. 
\end{lem}
\begin{proof}
Since $|A-LC|$ is a nonnegative matrix, we can apply Theorem 15 in \cite[p. 41]{Farina00-Book} to state that $\rho(\left|A-LC\right|)<1$ if and only if 
there exists a diagonal and positive definite matrix $P$ such that
$$|A-LC|^\top P|A-LC|-P\prec 0, $$
a condition which can be rewritten as
\begin{equation}\label{eq:PA-YC}
	|PA-YC|^\top P^{-1}|PA-YC|-P\prec 0. 
\end{equation}
where $Y=PL$. Now by posing $X=|PA-YC|$ and making use of the Schur complement trick, we see that (a) $\Rightarrow$ (b).  \\
Let us now prove that (b) $\Rightarrow$ (a). It is clear that if (b) holds then the first equation of condition \eqref{eq:LMI} implies, by the Schur complement rule,   that 
$X^\top P^{-1}X-P=(P^{-1}X)^\top P(P^{-1}X)-P\prec 0$ which  in turn means that $P^{-1}X$ is a Schur matrix, i.e.,  $\rho(P^{-1}X)<1$. Moreover,   it follows from the last condition of \eqref{eq:LMI},  the nonnegativity of $P$ and its diagonal structure (upon  multiplying each side on the left by $P^{-1}$), that $0\leq P^{-1}|PA-YC|=|A-LC|\leq P^{-1}X$. By applying  Lemma \ref{lem:stability-implication}, we conclude that $\rho(|A-LC|)\leq \rho(P^{-1}X)<1$. 
\end{proof}
Lemma \ref{lem:LMI} shows that one can compute the observer gain $L$ efficiently by solving a feasibility problem which is expressible in terms of Linear Matrix Inequalities (LMI) \cite{Boyd97-LMI-Book}. 

\noindent\rev{In comparison to some other results \cite{Wang18-SCL,Dinh19-IJC,Efimov16-ARC}, we do not put any positivity constraint on $A-LC$ for the existence of $L$ hence, yielding less conservative design conditions. In effect, 
if we introduce the notations  $\mathcal{L}_{\text{pos}}=\left\{L\in \Re^{n\times n_y}: \rho(A-LC)<1, A-LC\geq 0\right\}$ and $\mathcal{L}_{\text{abs}}=\left\{L\in \Re^{n\times n_y}: \rho(|A-LC|)<1\right\}$ for a given pair $(A,C)$, then clearly,  $\mathcal{L}_{\text{pos}}\subset \mathcal{L}_{\text{abs}}$. As a result, a search of the matrix gain $L$ over $\mathcal{L}_{\text{abs}}$ is more likely to be successful than a search over $\mathcal{L}_{\text{pos}}$.  
}

\begin{rem}
In addition to ensuring stability as required by Theorem \ref{thm:Gain-L-stability}, the matrix gain $L$ could be selected so as to achieve a certain level of convergence speed or to optimize a certain performance index, expressed for example in term of a norm of the system $(p_x(0),p_s)\mapsto {p}_{x,q}^{\cl}$ defined in \eqref{eq:phat-CL}. 
Such an approach was adopted in, for example, \cite{Rami08-CDC,Wang18-SCL,Briat16-Automatica}. Here however, our focus is rather on  characterizing tightness in the sense of Definition \ref{def:smallest-interval} for a given matrix gain $L$. 
\end{rem}

\section{Application to switched linear systems}\label{sec:sls}
We consider now applying the estimation  method discussed earlier to switched linear systems described by equations of the form
\begin{equation}\label{eq:SLS}
	\begin{aligned}
		&x(t+1)=A_{\sigma(t)}x(t)+B_{\sigma(t)}w(t)  \\
		&y(t)= C_{\sigma(t)}x(t)+v(t), 
	\end{aligned}
\end{equation}
where $(x,y,w,v)$ have the same significance as in \eqref{eq:LTI},  $\sigma:\mathbb{Z}_+\rightarrow\mathbb{S}$ with $\mathbb{S}=\left\{1,\ldots,s\right\}$ being a finite set, is the switching signal and  $(A_i,B_i,C_i)$, $i\in \mathbb{S}$, are the system matrices. We will consider that $(w,v)$ are still subject to Assumption \ref{assum:Bounding}.

\subsection{Settings for the state estimation}

\begin{assumption}\label{assum:switching-signal}
The switching signal $\sigma$ involved in \eqref{eq:SLS} is taken to be arbitrary but known. 
\end{assumption}

The first step of the estimator design method is to observe that for any set of matrices $\left\{L_i\in \Re^{n\times n_y}:i\in \mathbb{S}\right\}$, the state of the switched linear system \eqref{eq:SLS} obeys
\begin{equation}\label{eq:observer-switched}
	x(t+1) =F_{\sigma(t)}x(t)+G_{\sigma(t)}s(t)
\end{equation}
where $\sigma$ is the same switching signal as in \eqref{eq:SLS}, $F_i=A_i-L_iC_i$, $G_i=\bbm B_i & L_i & -L_i\eem$ for $i\in \mathbb{S}$ and  $s(t) = \bbm  w(t)^\top  & y(t)^\top & v(t)^\top\eem^\top$. 
Let $\Phi$ denote the transition matrix function of \eqref{eq:observer-switched} defined for all $(t,t_0)$ with $t\geq t_0$,  by
$$\Phi(t,t_0,\sigma)=\left\{\begin{array}{lll}I & & t=t_0\\F_{\sigma(t-1)}\cdots F_{\sigma(t_0)} & & t>t_0\end{array}\right.$$  
Then the state $x(t)$ of \eqref{eq:SLS} can be expressed as 
\begin{equation}\label{eq:x-sls}
	x(t)=\Phi(t,0,\sigma)x(0)+\sum_{j=0}^{t-1}\Phi(t,j+1,\sigma)G_{\sigma(j)}s(j).
\end{equation}
The second ingredient is an appropriate concept of stability on the homogenous part of \eqref{eq:observer-switched} defined by $z(t+1)=F_{\sigma(t)}z(t)$. In view of Assumption \ref{assum:switching-signal},  this notion of stability must hold regardless of the switching signal. It is therefore merely  a property of the finite set $\Sigma_F\triangleq \left\{F_i\in \Re^{n\times n}:  i\in \mathbb{S}\right\}$ of square matrices, hence the following definition. 
\begin{definition}[Stability of a finite set of matrices]\label{def:exp-stability}
The homogenous part of the discrete-time SLS \eqref{eq:observer-switched} (or equivalently, the finite collection  $\Sigma_F$ of matrices) is called 
\begin{itemize}
	\item  \emph{uniformly stable} if there is $c>0$ such that  for all $q\geq 1$ and for all $(i_1,\cdots,i_q)\in \mathbb{S}^q$, $\|F_{i_1}\cdots F_{i_q}\|_2\leq c $
	\item \emph{uniformly exponentially stable (u.e.s.)} if there exist some real numbers $c>0$ and $\lambda\in \interval[open]{0}{1}$  such that for all $q\geq 1$ and for all $(i_1,\cdots,i_q)\in \mathbb{S}^q$, 
	$\|F_{i_1}\cdots F_{i_q}\|_2\leq c\lambda^q $. 
\end{itemize}
\end{definition} 
An algebraic characterization of the stability of a switched linear system (or equivalently, the stability of a finite set of square  matrices) was  obtained in \cite{Sun11-Book} in terms of the joint spectral radius as follows. 
\begin{thm}[\cite{Sun11-Book}]\label{thm:algrebraic-cond-stability}
Let  $\Sigma=\left\{A_1,\ldots,A_m\right\}\subset \Re^{n\times n}$ be a finite collection of square matrices. Then 
$\Sigma$ is u.e.s. if and only if $\rho(\Sigma)<1$.  
\end{thm}
\noindent By Lemma \ref{lem:rho}, a direct corollary of this theorem is that there is an equivalence between stability of $|\Sigma|$ and that of $\Psi(\Sigma)$. Also, $\Sigma$ is stable whenever $|\Sigma|$ is stable.  

Now,  applying Lemma \ref{lem:Az+w} to \eqref{eq:x-sls} shows that for given gains $L_i$, the pair $(c_x^{\mbox{\tiny sls}},p_x^{\mbox{\tiny sls}})$ given by 
\begin{align}
&\!\!c_x^{\mbox{\tiny sls}}(t)\!=\!\Phi(t,0,\sigma)c_x(0)\!+\sum_{j=0}^{t-1}\Phi(t,j+1,\sigma)G_{\sigma(j)}c_s(j) \label{eq:cx-switched}\\
& \!\! p_x^{\mbox{\tiny sls}}(t)\!=\!|\Phi(t,0,\sigma)|p_x(0)\!+\!\sum_{j=0}^{t-1}|\Phi(t,j+1,\sigma)G_{\sigma(j)}|p_s(j) \label{eq:px-switched}
\end{align}

\noindent defines the tightest interval-valued estimator for the SLS \eqref{eq:SLS}   provided that the BIBO stability condition of Definition \ref{def:Interval-Estimator}  is satisfied. And it can be easily seen that the dynamic systems defined by \eqref{eq:cx-switched} and \eqref{eq:px-switched} are BIBO stable  if $\Sigma_F$ is u.e.s. Note that, like in Theorem \ref{thm:Tightest-CL}, we can also optimize over all sets $\left\{L_i\right\}_{i\in \mathbb{S}}$ of gains that render $\Sigma_F$ stable. \\
Concerning the implementation aspects, it is useful to observe that the function $c_x^{\mbox{\tiny sls}}$ in \eqref{eq:cx-switched} satisfies the one step-ahead equation 
\begin{equation}\label{eq:Realization-cx-SLS}
	c_x^{\mbox{\tiny sls}}(t+1)=F_{\sigma(t)}c_x^{\mbox{\tiny sls}}(t)+G_{\sigma(t)}c_s(t), \quad t\in \mathbb{Z}_+
\end{equation}
with $c_x^{\mbox{\tiny sls}}(0)=c_x(0)$. 
In contrast, realizing $p_x^{\mbox{\tiny sls}}$ in finite dimension is, like in the case of linear systems, still a challenging problem.  Nevertheless, a time-invariant switched linear state-space realization of \eqref{eq:px-switched} can, under certain conditions, be obtained by resorting to the realization theory of switched linear systems presented in \cite{Petreczky13-Automatica} but we will not elaborate more on this problem here. Turning instead to over-approximations of the estimator, it is interesting to see that  the truncated estimate discussed in Section \ref{subsub:Truncated} is applicable here as well. In particular, when the truncation order $q$ is equal to $1$, we obtain 
\begin{align}
&\hat{c}_x^{\mbox{\tiny sls}}(t+1)=F_{\sigma(t)}\hat{c}_x^{\mbox{\tiny sls}}(t)+ G_{\sigma(t)}p_s(t)\label{eq:cx-hat-sls} \\
	&\hat{p}_x^{\mbox{\tiny sls}}(t+1)=|F_{\sigma(t)}|\hat{p}_x^{\mbox{\tiny sls}}(t)+ |G_{\sigma(t)}|p_s(t) \label{eq:px-hat-sls}
\end{align}
with $\hat{c}_x^{\mbox{\tiny sls}}(0)=c_x(0)$ and  $\hat{p}_x^{\mbox{\tiny sls}}(0)=p_x(0)$. 
Thanks to the statement \eqref{eq:RhoSigma} of Lemma \ref{lem:rho}  and Theorem \ref{thm:algrebraic-cond-stability}, these latter equations define an interval estimator if 
\begin{equation}\label{eq:Sigma-absF}
	|\Sigma_{F}|\triangleq \left\{|F_i|=|A_i-L_iC_i|\in \Re^{n\times n}:  i\in \mathbb{S}\right\}
\end{equation}
is u.e.s. in the sense of Definition \ref{def:exp-stability}. 
Hence the question we discuss next is how to select the matrix gains $L_i\in \Re^{n\times n_y}$ so as to render $|\Sigma_{F}|$ u.e.s., hence making \eqref{eq:cx-hat-sls}-\eqref{eq:px-hat-sls} a valid interval estimator.  Observe that $|\Sigma_{F}|$ is a  positive  discrete-time switched linear system just as those  studied in \cite{Blanchini15-FTSC} (in continuous-time) considered under arbitrary switching signal. 

\subsection{Guaranteeing the stability condition}
In this section, we derive a  tractable condition for computing effectively  gains $L_i$ which ensure exponential stability of $|\Sigma_{F}|$. For this purpose, we will need the following lemma. 
\begin{lem}\label{lem:stability-implication-v2}
Consider two finite collections of nonnegative matrices $\Sigma=\left\{A_i\in \Re_+^{n\times n}:  i\in \mathbb{S}\right\}$ and $\overline{\Sigma}=\left\{\bar{A}_i\in \Re_+^{n\times n}:   i\in \overline{\mathbb{S}}\right\}$, where $\mathbb{S}$ and $\overline{\mathbb{S}}$ are finite sets with possibly different cardinalities. 
If for any $i\in \mathbb{S}$, there is $j\in \overline{\mathbb{S}}$ such that $A_i\leq \bar{A}_{j}$, then $\Sigma$ is uniformly stable (resp. u.e.s.) whenever $\overline{\Sigma}$ is uniformly stable (resp. u.e.s.). 
\end{lem}
\begin{proof}
This is a direct consequence of Lemma \ref{lem:stability-implication} and Theorem \ref{thm:algrebraic-cond-stability}.   
\end{proof}
\noindent It also follows naturally that if $\Sigma$ is u.e.s., then so is any non empty subset of $\Sigma$.  \\
Now the main result of this section can be stated as follows. 
\begin{thm} \label{thm:stability-estimator-SLS}
The following chain of implications hold: $(A)\Rightarrow (B) \Rightarrow (C)$ where
(A), (B), (C) correspond to the following statements:
\begin{enumerate}
\item[(A)] There exist some matrices $L_i\in \Re^{n\times n_y}$, $i\in \mathbb{S}$, 
	 \rev{diagonal matrices $P_i\in \Re_+^{n\times n}$ satisfying  $P_i\succ 0$, $i\in \mathbb{S}$}, and a strictly positive number $\gamma>0$ such that 
	\begin{equation}\label{eq:SLS-Stability-LMI}
		|A_i-L_iC_i|^\top P_j|A_i-L_iC_i|-P_i\preceq -\gamma I 
	\end{equation}
	for all $(i,j)\in \mathbb{S}^2$. 
	\item[(B)] There exist some matrices  $X_{ji}\in \Re_+^{n\times n}$,  $Y_{ji}\in \Re^{n\times n_y}$, $(i,j)\in \mathbb{S}^2$, some diagonal positive-definite matrices $\Lambda_i\in\Re^{n\times n}$, $i\in \mathbb{S}$, and a real number $\eta>0$ such that
\begin{align}
		&\BM \Lambda_j & X_{ji}\\ X_{ji}^\top & \: \: \Lambda_i-\eta I \EM\succeq 0 \label{eq:LMI-SLS}\\
		& \left|\Lambda_j A_i-Y_{ji}C_i\right|\leq X_{ji} \label{eq:Ineq-SLS}
\end{align}
for all $(i,j)\in \mathbb{S}^2$. 
\item[(C)] $|\Sigma_{F}|$ in \eqref{eq:Sigma-absF} is u.e.s. with $L_i=\Lambda_i^{-1}Y_{ii}$, $i\in \mathbb{S}$. 
\end{enumerate}
\end{thm}
\begin{proof}
\textit{(A) $\Rightarrow$ (B):} If  (A) holds, then by using the facts that $P_j$ is diagonal, nonnegative and nonsingular, we can rewrite condition \eqref{eq:SLS-Stability-LMI} as 
$$|P_j(A_i-L_iC_i)|^\top P_j^{-1}|P_j(A_i-L_iC_i)|-P_i +\gamma I\preceq 0.$$
Now by setting $\eta=\gamma$, $Y_{ji}=P_jL_i$, $X_{ji}=|P_jA_i-Y_{ji}C_i)|$, $\Lambda_i=P_i$ and calling upon the Schur complement rule, the statement (B) follows.

\textit{(B) $\Rightarrow$ (C):}
Assume that condition (B) holds. Then, by Eq.  \eqref{eq:LMI-SLS}, we have, upon applying the Schur complement rule,  $\Lambda_i-\eta I-X_{ji}^\top \Lambda_j^{-1}X_{ji}\succeq 0$ for all $(i,j)\in\mathbb{S}^2$. Writing this in the form $$\Lambda_i-\eta I-(\Lambda_j^{-1}X_{ji})^\top \Lambda_j(\Lambda_j^{-1}X_{ji})\succeq 0$$
 reveals, by application of Theorem 28 in \cite[p.267]{Vidyasagar02-Book}, that indeed $\Sigma'\triangleq \left\{\Lambda_j^{-1}X_{ji}: (i,j)\in \mathbb{S}^2\right\}$ is exponentially stable in the sense of Definition \ref{def:exp-stability}. To see this, it suffices to consider the system defined by $z(t+1)=\tilde{A}(t)z(t)$ with $\tilde{A}(t)\in \Sigma'$ for all time index $t\in \mathbb{Z}_+$. Then the function $V:\mathbb{Z}_+\times \Re^n$ defined by $V(t,x)=x^\top \Lambda_j x$ whenever $\tilde A(t-1)= \Lambda_j^{-1}X_{ji}$, $t\geq 1$ and  $V(0,x)=x^\top \Lambda_{j_0} x$ for an arbitrary $j_0\in \mathbb{S}$,  satisfies all the conditions of the theorem cited above. \\
 On the other hand, Eq. \eqref{eq:Ineq-SLS} imply that $\Lambda_j|A_i-\Lambda_j^{-1}Y_{ji}C_i|\leq X_{ji}$ since $\Lambda_j$ is positive and diagonal. This in turn implies that $|A_i-L_{ji}C_i|\leq \Lambda_j^{-1}X_{ji}$ where $L_{ji}=\Lambda_j^{-1}Y_{ji}$ for all $(i,j)\in \mathbb{S}^2$. Hence by applying Lemma \ref{lem:stability-implication} (observe that $|\Sigma_{F}|\subset \Sigma'$ if $L_i$ is taken to be equal to $L_{ii}=\Lambda_i^{-1}Y_{ii}$), we can conclude that $|\Sigma_{F}|$ is exponentially stable. 
\end{proof}
Theorem \ref{thm:stability-estimator-SLS} shows  that the problem of designing an interval-valued estimator in the form \eqref{eq:cx-hat-sls}-\eqref{eq:px-hat-sls} for the SLS system \eqref{eq:SLS} can be relaxed to the problem of solving the convex feasibility problem \eqref{eq:LMI-SLS}-\eqref{eq:Ineq-SLS} for the matrix gains $L_{ii}\in \Re^{n\times n_y}$, $i\in \mathbb{S}$. Hence a numerical solution can be efficiently obtained by relying on existing semi-definite programming solvers. 

\begin{rem}
Similarly as in the case of LTI systems (see Theorem \ref{thm:Gain-L-stability}), we can construct a $q$-order truncated interval estimator in the switched system case as well. In effect, it can be shown that for $\Sigma_F=\left\{F_i:i\in \mathbb{S}\right\}$, $\rho(\Sigma_F)<1$ implies the existence of some $q^\star$ such that $\rho(|\Sigma_F^q|)<1$ for all $q\geq q^\star$ with the notation $|\Sigma_F^q|$ defined by $|\Sigma_F^q|=\left\{|U_1\cdots U_q|: U_i\in \Sigma_F, i=1,\ldots,q\right\}$ (i.e., $|\Sigma_F^q|$ is the set of matrices formed as the product of $q$ members of $\Sigma_F$, taken in absolute value. Based on this fact, we can search for the $\left\{L_i\right\}$ by dropping the absolute values in the stability condition \eqref{eq:SLS-Stability-LMI}. The associated LMI takes the form
$$\begin{bmatrix}P_j & P_jA_i-Y_{ji}C_i\\ (P_jA_i-Y_{ji}C_i)^\top & P_i-\gamma I\end{bmatrix}\succeq 0$$
with variables $Y_{ji}\in \Re^{n\times n_y}$, $(i,j)\in \mathbb{S}^2$, $P_i\in \Re^{n\times n}$ (positive-definite and no longer diagonal), $i\in \mathbb{S}$,  $\gamma>0$. \\
 \noindent By then selecting an appropriate value for $q$,  the pair $(c_x^{\sls},p_{x,q}^{\sls}):\mathbb{Z}_+\rightarrow\Re^n\times \Re_+^n$ defined respectively in \eqref{eq:Realization-cx-SLS} and \eqref{eq:phat-CL-SLS} (see below)  constitute an interval-valued estimator for the switched system \eqref{eq:SLS}, 
\begin{equation}\label{eq:phat-CL-SLS}
	{p}_{x,q}^{\sls}(t)=
\left\{	\begin{array}{ll}
\Phi(t,0,\sigma)|p_x(0)\!+\!\sum_{j=0}^{t-1}|\Phi(t,j+1,\sigma)G_{\sigma(j)}|p_s(j)  &  \mbox{ if } t=0,\ldots,q\\
\Phi(t,t-q,\sigma)|p_{x,q}^{\sls}(t-q)\!+\!\sum_{j=t-q}^{t-1}|\Phi(t,j+1,\sigma)G_{\sigma(j)}|p_s(j), & \mbox{ if } t> q
	 \end{array}\right.
\end{equation}
The advantage of such an estimator is that its (relaxed) synthesis conditions (analogues of \eqref{eq:LMI-SLS}-\eqref{eq:Ineq-SLS}) are less conservative. 
Its inconvenience is an increase of computational cost in the numerical evaluation of $p_{x,q}^{\sls}$ in case $q$ is taken large.  
\end{rem}

\section{Numerical results}\label{sec:Simulations}
This section reports some simulation results that illustrate the performances of the interval estimators discussed in the paper. 
\subsection{An example of linear system in open loop}\label{subsec:LTI-open-loop}
We first consider a linear system in the  form \eqref{eq:LTI} in an open-loop configuration, that is, without making use of any measurement. 
In order to be able to implement all the estimators discussed in Section \ref{sec:Open-Loop}, the dynamic matrix  $A$ is selected such that $\rho(|A|)<1$, 
\begin{equation}\label{eq:Example-LTI} 
A=\BM 0.10 & 0.60 &  0.05\\
   0.20  & 0.35  & -0.50\\
   -0.55 &  -0.15 & 0.40\EM, 
\quad 
B = \BM -0.50\\0.70\\1\EM. 
\end{equation}
The set of initial states is defined by  $c_x(0)=\bbm 0.50 & -1 & -2\eem^\top$,  $p_x(0)=\bbm 3& 2& 4\eem^\top$.  
For any $t\in \mathbb{Z}_+$, we let the input intervals be defined by $c_w(t)=\sin(2\pi\nu_c t)$ and $p_w(t)=0.10|\cos(2\pi\nu_p t)|$ with $\nu_c=0.01$ and $\nu_p=0.001$.

We first check the realizability of $p_x$ in \eqref{eq:px(t)} in state-space form. 
For the example \eqref{eq:Example-LTI}, it is numerically found that the Hankel matrix $\mathcal{H}_{i,j}$ defined in \eqref{eq:Hankel} has finite constant rank equal to $6$ for sufficiently large $i$ and $j$. It therefore follows from Theorem \ref{thm:realization} that the tightest interval-valued estimator given by \eqref{eq:cx(t)}-\eqref{eq:px(t)} admits an LTI state-space realization $\left(\mathcal{A},\mathcal{B},\mathcal{C},\phi_0\right)$ as defined in Section \ref{subsec:realization} with minimal dimension $d=6$. Hence it can be cheaply implemented.   

With these data the estimators  defined in Eqs \eqref{eq:px(t)},  \eqref{eq:realization-px}, \eqref{eq:phat}  and \eqref{eq:phat1} are simulated. More precisely, $100$ possible state trajectories are obtained from inputs $\left\{w(t)\right\}$ and initial states $x(0)$ drawn  randomly from the corresponding intervals defined above. Figure \ref{fig:simulation_state} shows that all the estimators enclose the true state trajectories in gray. As proved in the paper,  \eqref{eq:px(t)} yields the smallest interval estimator. It is interesting to observe that the estimator \eqref{eq:phat} (which is implemented here for $q=1$ and $q=2$) can provide an estimate that is very close to the best one without $q$ being large. The estimate delivered by the estimator \eqref{eq:realization-px} (with here $r^o=0.3\geq p_w(t)$ for all $t$; see Section \ref{subsec:Approximation} for a definition of $r^o$) is a little worse but the worst of all on this example  is the result returned by \eqref{eq:phat1} which gives a quite large interval set.    
\begin{figure}[h!]
	\centering
	\psfrag{time}[][]{\tiny time samples}
	\psfrag{x}[][]{\scriptsize $x_1$} %
	\psfrag{xbar}{\tiny Tightest} 
	\psfrag{xbar1}{\tiny $q=1$} 
	\psfrag{xbar2}{\tiny $q=2$}
	\psfrag{xbar3}{\tiny app.$p_w$} 
\includegraphics[width=9cm,height=7cm]{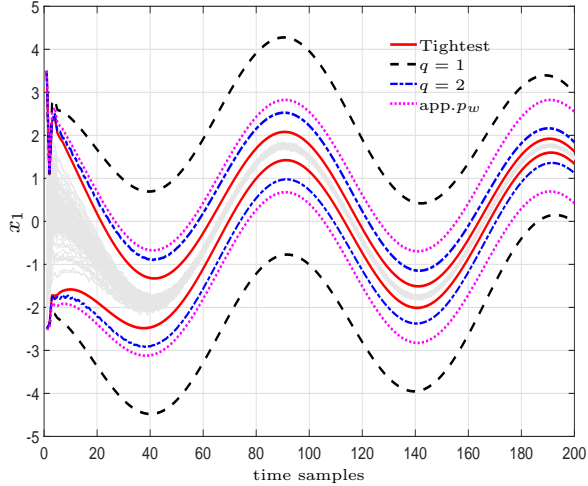}
\caption{Comparison of open-loop interval estimators: In red the tightest estimator defined in Eq. \eqref{eq:px(t)}; Truncated estimator Eq.\eqref{eq:phat} in black (for $q=1$) and in blue (for $q=2$);  the estimator Eq.\eqref{eq:realization-px} based on an over-approximation of $p_w$ in magenta. Only the first state components are represented here. }
	\label{fig:simulation_state}
\end{figure}

\subsection{An example of switched linear system}
We now consider an example of switched linear system  in the form \eqref{eq:SLS} with matrices $(A_i,B_i,C_i)$  defined as follows:
$$
\begin{array}{ccc}
A_1 =\BM -0.40  &  0.075  & -0.55\\
   -0.50  & -0.15 &  0.50\\
   -0.16  &  0.75 & 0.45\EM, 
&B_1 =\BM -0.60\\-1.20\\0.25\EM,  
& C_1^\top  =\BM 0 \\ -0.85 \\ -1\EM 
\end{array}
$$
$$
\begin{array}{ccc}
A_2 =\BM -0.30 & -0.20  &  0.50\\
          -0.25  & -0.80  & -0.15\\
          -0.45 & 0.6 & -0.25\EM, 
&B_2 =\BM 0.20\\-0.25 \\-1\EM,
& C_2^\top =\BM 0.50\\  0 \\0.15 \EM 
\end{array}
$$
$$
\begin{array}{ccc}
A_3 =\BM 0.25  & -0.70 &   0.15\\
          0.06 &-0.10 &  -0.70\\
          0.80  & 0.60 &0.15\EM, 
&B_3 =\BM 0\\ 0.40 \\1.85\EM, 
&C_3^\top =\BM 0.20 \\ -0.06\\ 2\EM
\end{array}
$$
The simulation is carried out for the closed-loop scenario in the following setting: the set of initial conditions and the set of  admissible input signals are kept the same as in Section \ref{subsec:LTI-open-loop}. As to the measurement noise $\left\{v(t)\right\}$, it is assumed to live in the constant interval $\interval{-0.1}{0.1}$. We then solve \eqref{eq:LMI-SLS}-\eqref{eq:Ineq-SLS} to find observer stabilizing gains $L_i$ and plug them in \eqref{eq:cx-hat-sls}-\eqref{eq:px-hat-sls}. Finally, estimating the state trajectory of the SLS example described above using the tightest estimator \eqref{eq:cx-switched}-\eqref{eq:px-switched} and the one in \eqref{eq:cx-hat-sls}-\eqref{eq:px-hat-sls} gives the result depicted in Figure \ref{fig:state-estimate-SLS}. Here the switching signal is piecewise constant with dwell time of $30$ time samples in each mode.  Again, it can be noticed that the actual state is effectively enclosed by both estimators and that the claimed property of tightness is supported by the empirical evidence.   
\begin{figure}[h!]
	\centering
	\psfrag{time}[][]{\tiny time samples}
	\psfrag{x}[][]{\scriptsize $x_1$} %
	\psfrag{xbarT}{\tiny Tightest} 
	\psfrag{xbar}{\tiny Truncated ($q=1$)} 
	\psfrag{xact}{\tiny Actual $x_1$}
\includegraphics[width=9cm,height=7cm]{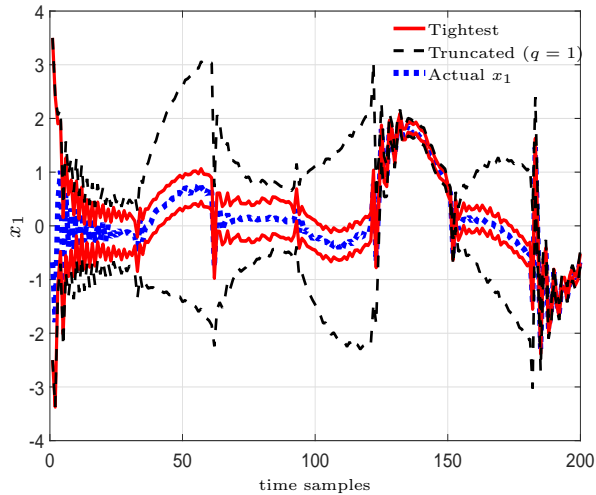}
\caption{Estimation in closed-loop for a SLS: Actual state trajectory resulting from a single simulation (blue dotted); the tightest estimator (red solid) and estimator in observer form \eqref{eq:cx-hat-sls}-\eqref{eq:px-hat-sls} (black dashed). }
\label{fig:state-estimate-SLS}
\end{figure}
\section{Conclusion}\label{sec:Conclusion}
In this paper we have presented a new approach to the interval-valued state estimation problem. The proposed framework is mainly discussed for the case of discrete-time linear systems and later, extended to switched linear systems. In particular, we have derived the tightest interval estimator which enclose all the possible state trajectories generated by discrete-time linear (and switched linear) systems. Such an estimator can, under some conditions, be realized in an LTI state-space form. When this condition fails to hold, alternative over-approximations can be considered. Considering one of those approximations in the form of one-step ahead state-space representation, we propose a method to compute the parameters of the estimator by solving convex feasibility problems. 

\bibliographystyle{abbrv}

\end{document}